\documentclass{llncs}
\usepackage{llncsdoc}

\usepackage{url}
\usepackage{amsmath}
\usepackage{amssymb}
\usepackage{newalg}
\usepackage{clrscode}
\usepackage{graphicx}
\usepackage{setspace}
\usepackage{cite}

\newtheorem{prb}{Problem}
\newtheorem{thm}{Theorem}
\newtheorem{prop}{Property}
\newtheorem{defn}{Definition}

\title{Linear Time Inference of Strings from Cover Arrays using a Binary Alphabet}
\author{Tanaeem M. Moosa\thanks{Currently working at Google Inc., USA.} \and Sumaiya Nazeen \and M. Sohel Rahman \and Rezwana Reaz}

\institute{A$\ell$EDA Group\\
Department of CSE, BUET\\
Dhaka-1000, Bangladesh\\
\email{\{tanaeem,nazeen,msrahman,rimpi\}@cse.buet.ac.bd}}

\begin{document}

\maketitle

\begin{abstract}
Covers being one of the most popular form of regularities in strings, have drawn much attention over time. In this paper, we focus on the problem of linear time inference of strings from cover arrays using the least sized alphabet possible. We present an algorithm that can reconstruct a string $x$ over a two-letter alphabet whenever a valid cover array $C$ is given as an input. This algorithm uses several interesting combinatorial properties of cover arrays and an interesting relation between border array and cover array to achieve this. Our algorithm runs in linear time.
\end{abstract}

\section{Introduction}
\label{sec:intro}
A substring $w$ of string $x$ is called a {\itshape cover} of $x$ if $x$ can be constructed by concatenation and/or superposition of $w$. Though $x$ is always a cover of itself, we do not consider so, in this paper. We limit our focus on the so-called {\itshape aligned covers} where the cover $w$ needs to be a proper substring and also a {\itshape border} (i.e., a prefix and a suffix) of $x$. For example, the string $x = abcababcabcabcab$ is constructed by the concatenation~(at position 6) and superposition~(at positions 9 and 12) of $w = abcab$. Thus $x$ has a proper cover, $w$ which is also a border. A string that has a proper cover is called \emph{coverable} or \emph{quasiperiodic}, otherwise it is \emph{superprimitive}~\cite{DBLP:journals/tcs/ApostolicoE93}. The array $C$ is called the \emph{minimal-cover} (resp. \emph{maximal-cover}) \emph{array} of the string $x$ of length $n$, if for each $i, 1 \leq i \leq n$, $C[i]$ stores either the length of the shortest (resp. longest) cover of $x[1 \twodots i]$, when such a cover exists, or zero otherwise. The array $B[1\twodots n]$ is the \emph{border array} of the string $x$ if $B[i]$ stores the length of the longest border of $x[1\twodots i]$, $1\leq i\leq n$.

 Repetitions in strings like periods, borders, covers etc. have always been a subject of great interest for the computer scientists because of its diverse applications in fields like molecular biology, probability theory, coding theory, data compression and formal language theory. In fact, in the last two decades string periodicity has drawn a lot of attention from different disciplines of science. The famous \proc{KMP}~\cite{DBLP:journals/siamcomp/KnuthMP77} pattern matching algorithm depends on the \emph{failure function} which is nothing but the \emph{border array}. Another well-known pattern matching algorithm namely the \proc{Boyer-Moore} algorithm~\cite{DBLP:journals/cacm/BoyerM77} makes use of similar kind of repetitions in strings. Such repetitions in strings are often encoded in data structures like graphs and integer arrays~\cite{DBLP:conf/mfcs/BannaiIST03}. Thus, researchers have shown interest not only in finding repetitions in strings but also in reconstructing strings from those repetitive information. Apostolico \emph{et al.}~\cite{DBLP:conf/birthday/ApostolicoB97} gave an online linear runtime algorithm computing the minimal-cover array of a string. Smyth \emph{et al.}~\cite{DBLP:journals/algorithmica/LiS02} provided an online linear runtime algorithm for computing the maximal cover array which describes all the covers of a string. The problem of reverse engineering a string was first introduced by Fran{\v e}k {\itshape et al.}~\cite{FranekLRSSY02}. They proposed a method to check if an integer array is a \emph{border array} for some string. Border arrays are better known as {\itshape failure functions}~\cite{DBLP:books/aw/AhoHU74}. They showed an online linear time algorithm to verify if a given integer array is a border array for some string $w$ on an unbounded alphabet. Duval {\itshape et al.}~\cite{DBLP:journals/jalc/DuvalLL05} gave an online linear time algorithm for bounded alphabet to solve this problem. Bannai \emph{et al.}~\cite{DBLP:conf/mfcs/BannaiIST03} solved the problem of inferring a string from a given suffix array on minimal sized alphabet by proposing a linear time algorithm. Smyth \emph{et al.} discussed a possible solution of string inference problem from prefix arrays in~\cite{DBLP:conf/spire/SmythW08}.

Crochemore \emph{et al.}~\cite{DBLP:conf/cpm/CrochemoreIPT10} presented a constructive algorithm checking if an integer array is the minimal-cover or maximal-cover array of some string. When the array is valid, their algorithm produces a string over an unbounded alphabet whose cover array is the input array. All these algorithms run in linear time. Very recently, Tomohiro \emph{et al.}~\cite{DBLP:conf/lata/IIBT09} proposed a way to verify whether a given integer array is a valid parameterized border array (p-border array) for a binary alphabet. They further extended their work in~\cite{DBLP:conf/cpm/IIBT10} by giving an ${O}(n^{1.5})$-time ${O}(n)$-space algorithm to verify if a given integer array of length $n$ is a valid p-border array for an unbounded alphabet.

In this paper, we address the open problem stated in~\cite{DBLP:conf/cpm/CrochemoreIPT10}. We present a linear time algorithm for reconstruction of a string from cover array using least sized alphabet. Our algorithm is closely analogous to the \proc{MinArrayToString} algorithm in~\cite{DBLP:conf/cpm/CrochemoreIPT10}. We achieve the least possible size of alphabet by incorporating an interesting relation between border array and cover array of a string presented in~\cite{DBLP:journals/algorithmica/LiS02}. In fact, our algorithm is able to reconstruct strings from valid cover arrays using an alphabet consisting of no more than two characters.

The rest of this paper is organized as follows. Section~\ref{sec:pre} gives an account of definitions and notations used throughout the paper. Section~\ref{sec:prb} presents the addressed problem formally and lists important properties and lemmas used later. In Section~\ref{sec:alg} we describe our algorithm and main findings. Section~\ref{sec:exp} provides some experimental analysis of our algorithm. Finally, Section~\ref{sec:con} gives the conclusions.

\section{Preliminaries}
\label{sec:pre}
A string\index{string} $x$ is a finite sequence of symbols drawn from an alphabet\index{alphabet} $\Sigma$, where $\Sigma[i]$ denotes the $i$-th symbol of $\Sigma$. The set of all strings over $\Sigma$ is denoted by $\Sigma^{*}$. The $length$ of a string\index{string!length} is denoted by $|x|$. The {\itshape empty string}\index{string!empty}, the string of length zero, is denoted by $\epsilon$.

A string $w$ is a \emph{factor}\index{factor} of string $x$ if $x\ =\ uwv$ for two strings $u$ and $v$. It is a {\itshape prefix}\index{prefix} of $x$ if $u$ is empty and {\itshape suffix}\index{suffix} of $x$ if $v$ is empty. It is a {\itshape proper prefix}\index{prefix!proper} of $x\ =\ wv$ when v $is$ nonempty and a {\itshape proper suffix}\index{suffix!proper} of $x\ =\ uw$ when $u$ is nonempty. For example, $w\ =\ abc$ is a \emph{factor} of $x\ =\ pqabcmn$, a \emph{proper prefix} of $x\ =\ pqabc$ and a \emph{proper suffix} of $x =  abcmn
$, where $u\ =\ pq$, $v\ =\ mn$ and $w$, $u$, $v$, $x$ $\in$ $\Sigma^{*}$.

A string $u$ is a {\itshape period}\index{period} of $x$ if $x$ is a prefix of $u^{k}$ for some positive integer $k$, or equivalently if $x$ is a prefix of $ux$. The {\itshape period} of $x$ is the shortest period of $x$. For example, if $x = abcabcab$, then $abc$, $abcabc$ and the string $x$ itself are periods of $x$, while $abc$ is the \emph{period} of $x$.

A string u is a \emph{border}\index{border} of $x$ if $u$ is a \emph{prefix}\index{prefix} and a \emph{suffix}\index{suffix} of x and $u\ \neq\ x$. A \emph{border} $u$ of $x[1\twodots i]$ with $i>0$ has one of the two following forms:
\begin{itemize}
\item $u = \epsilon$
\item $u = x[1\twodots j]x[j+1]\ $ with $j+1 < i$ and where $x[1\twodots j]$ is a border of $x[1\twodots i-1]$ and $x[i] = x[j+1]$
\end{itemize}
Thus, a \emph{border}\index{border!regular} $u$ of a regular string $x = x[1\twodots n]$ is a proper \emph{prefix} of $x$ that is also a \emph{suffix} of $x$; thus $u = x[1\twodots b] = x[n-b+1\twodots n]$ for some $b \in 0\twodots n-1$.

The \emph{border array}\index{border array!regular} of a \emph{regular string} $x = x[1\twodots n]$ is an integer array $B = B[1\twodots n]$ such that, for every $i \in 1\twodots n,\ B[i]$ is the length of the longest border of $x[1\twodots i]$.

A string $w$ of length $m$ is a {\itshape cover}\index{cover} of string $x[1\twodots n]$ if both $m < n$ and there exists a set of positions $P \subseteq \{1,\ldots,n-m+1\}$ satisfying $x[i\twodots i+m-1] = w$ for all $i \in P$ and $\bigcup_{i\in P}\{i,\ldots,i+m-1\} = \{1,\ldots,n\}$. Therefore, if substring $w$ of string $x$ is a {\itshape cover}\index{cover} of $x$, then $x$ can be constructed by concatenation and/or superposition of $w$. Though $x$ is always a cover of itself, we do not consider so, in this paper. We limit our focus on the so-called \emph{aligned covers}\index{cover!aligned} where the cover $w$ needs to be a \emph{proper substring} and also a {\itshape border}\index{border}~(i.e., a prefix and a suffix) of $x$. For example, the string $x = abcababcabcabcab$ has proper cover $w = abcab$ which is also a border. A string that has a proper cover\index{cover!proper} is called \emph{coverable}\index{string!coverable} or \emph{quasiperiodic}\index{string!quasiperiodic}, otherwise it is \emph{superprimitive}\index{string!superprimitive}.

The {\itshape minimal-cover array}\index{cover array!minimal} $C$ of $x$ is the array of integers $C[1\twodots n]$ for which $C[i],1 \leq i \leq n$, stores the length of the shortest cover of the prefix $x[1\twodots i]$, if such a cover exists, or zero otherwise. The {\itshape maximal-cover array}\index{cover array!maximal} $C^{M}$ stores longest cover at each position instead. An example is given below. In what follows, we mean by \emph{cover array} $C$, the minimal cover array unless otherwise specified. An example of minimal and maximal cover array is given in Figure~\ref{fig:cover}.

\begin{figure}[!htbp]
\begin{center}
  \begin{tabular}[c]{r*{23}{c}c}
    $i $& &1&2&3&4&5&6&7&8&9&10&11&12&13&14&15&16&17&18&19&20&21&22&23\\
    \hline
    $x[i] $& &a&b&a&a&b&a&b&a&a&b&a&a&b&a&b&a&a&b&a&b&a&b&a\\
    $C[i] $& &0&0&0&0&0&3&0&3&0&5&3&0&5&3&0&3&0&5&3&0&3&0&3\\
    $C^{M}[i] $& &0&0&0&0&0&3&0&3&0&5&6&0&5&6&0&8&9&10&11&0&8&0&3\\
    \end{tabular}
\end{center}
\label{fig:cover}
\caption{Illustration of \emph{minimal} and \emph{maximal cover array}.}
\end{figure}

Adopting the graphical approach described in~\cite{DBLP:conf/cpm/CrochemoreIPT10}, we define the cover graph as follows:

\begin{defn}
\label{defn:d1}
    A cover graph $G=(V,E)$ is an undirected graph where $V = \{1\ldots,n\}$ and each vertex $i, 1 \leq i \leq n$ corresponds to index $i$ of string $x[1..n]$. The edge set E is defined as follows based on the {\itshape equivalence relation} of indices of $x$:$$E = \bigcup_{i = 1,\ldots,n}\bigcup_{j=1,\ldots,\gamma[i]}(j, i-\gamma[i]+j),$$ where $\gamma$ is any valid cover array.
\end{defn}

Figure~\ref{fig:coverg} shows a \emph{Cover Graph} constructed from given cover array $C$.
\begin{figure}[!htbp]
	\begin{center}
  	\begin{tabular}[c]{r*{14}{c}c}
	\hline
    $i $& &1&2&3&4&5&6&7&8&9&10&11&12&13&14\\
    \hline
    $C[i] $& &0&0&0&0&0&3&0&3&0&5&3&0&5&3\\
	\hline
    \end{tabular}
	\end{center}
	\begin{center}
		{\bf (a)}
	\end{center}
    \begin{center}
     \includegraphics[width=0.5\textwidth]{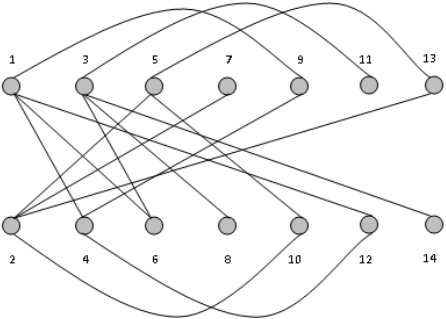}
    \end{center}
	\begin{center}
		{\bf (b)}
	\end{center}
	\caption[Illustration of a \emph{Cover Graph}.]{Illustration of a \emph{Cover Graph}. (a) Input cover array $C$, and (b) Corresponding \emph{Cover Graph}.}
    \label{fig:coverg}
\end{figure}

\section{Problem Definition \& Important Properties}
\label{sec:prb}
We start with a formal definition of the problem handled in this paper.

\begin{prb}
\label{prb:p1}
Linear time inference of strings using the least sized alphabet from cover arrays.
\end{prb}

{\bf Input:} A valid  cover array $C$, of length $n$.

{\bf Output:} A string $x$ of length $n$ on a minimum sized alphabet.
\\
\\
Before presenting our algorithm, we mention some important properties related to the cover array and border array which will be used later.

\begin{prop}[\emph{Transitivity property of a cover}~\cite{DBLP:conf/cpm/CrochemoreIPT10}]
\label{prop:cover1}
If each of $u$ and $v$ covers $x$ and $|u|<|v|$, then $u$ covers $v$.
\end{prop}

\begin{prop}[\emph{Totally covered position in cover array}\index{cover array!totally covered position}~\cite{DBLP:conf/cpm/CrochemoreIPT10}]
\label{prop:cover2}
A position $j \neq 0$ of a cover array $C$ is called {\itshape totally covered}, if there is a position $i > j$ for which $C[i] \neq 0$ and $i-C[i]+1 \leq j-C[j]+1<j$.
\end{prop}

\begin{prop}[\emph{Pruned minimal cover array}\index{cover array!pruned}~\cite{DBLP:conf/cpm/CrochemoreIPT10}]
\label{prop:cover3}
Let $C^{P}$ be obtained from $C$ by setting $C[i] = 0$ for all totally covered indices $i$ on $C$. We call $C^{P}$ the pruned minimal cover array of $x$ . Figure~\ref{fig:coverp} shows an example of pruned minimal cover array.

\begin{figure}[!htbp]
\begin{center}
  \begin{tabular}[c]{r*{24}{c}c}
    $i $& &1&2&3&4&5&6&7&8&9&10&11&12&13&14&15&16&17&18&19&20&21&22&23&24\\
    \hline
    $x[i] $& &a&b&a&a&b&a&b&a&a&b&a&b&a&a&b&a&a&b&a&b&a&a&b&a\\
    $C[i] $& &0&0&0&0&0&3&0&3&0&5&3&7&3&9&5&3&0&5&3&0&3&9&5&3\\
    $C^{P}[i] $& &0&0&0&0&0&3&0&0&0&0&0&0&0&9&5&0&0&0&0&0&0&9&5&3\\
  \end{tabular}
\end{center}
\label{fig:coverp}
\caption{Illustration of \emph{minimal} and \emph{pruned minimal cover array}.}
\end{figure}
\end{prop}

\begin{prop}[\emph{Border constraint on cover array}\cite{DBLP:conf/cpm/CrochemoreIPT10}]
\label{prop:cover4}
The nonzero values in $C$ induce an equivalence relation\index{equivalence relation|see{cover array!border constraint}} on the positions of every string that has the minimal-cover array $C$. More precisely, if we find the value $l \neq 0$ in position $i$ of $C$, then this imposes the constraints $$x[k] = x[i-l+k]$$ for $k = 1,\ldots,l$. The positions $k$ and $i-l+k$ are bidirectionally linked.
\end{prop}

\begin{prop}[\cite{DBLP:conf/cpm/CrochemoreIPT10}]
\label{prop:cover5}
Let $i$ and $j$ be positions such that $j<i,\ j-C[j] \geq i-C[i], C[i] \neq 0$ and $C[j] \neq 0$. Furthermore, let $r = j-(i-C[i]+1)$. If $i-C[i] = j-C[j]$, then $C[r] = 0$, otherwise if $i-C[i] < j-C[j]$, then $C[r]=C[j]$.
\end{prop}

\begin{prop}[\cite{DBLP:conf/cpm/CrochemoreIPT10}]
\label{prop:cover8}
Let $i$ and $j$ be positions such that $j < i$ and $j-C[j] < i-C[i]$. Then $(i-C[i])-(j-C[j]) > C[j]/2$.
\end{prop}

\begin{prop}[\cite{DBLP:conf/cpm/CrochemoreIPT10}]
\label{prop:cover6}
The sum of the elements of $C^{P}$ does not exceed $2n$.
\end{prop}

\begin{prop}[\cite{DBLP:journals/algorithmica/LiS02}]
\label{prop:cover7}
For every integer $i \in 1\twodots n-1$, if $B[i] \leq B[i-1]$, then $C[i] = 0$
\end{prop}

\section{Our Algorithm}
\label{sec:alg}

In this section, we present an efficient algorithm, which reconstructs a string $x$ from a cover array $C[1\twodots n]$ on a binary alphabet
in linear time. We call this algorithm, Algorithm \proc{SIMA}~({\bf S}tring {\bf I}nference using {\bf M}inimum-sized {\bf A}lphabet). We assume that a valid cover array will be given as input. The validity of a cover array can be easily checked by Property~\ref{prop:cover2}~\cite{DBLP:conf/cpm/CrochemoreIPT10} and Property~\ref{prop:cover8}~\cite{DBLP:conf/cpm/CrochemoreIPT10} using the same approach used in~\cite{DBLP:conf/cpm/CrochemoreIPT10} without changing the running time of our algorithm.

The algorithm uses the following arrays:
\begin{itemize}
\item $C[1\twodots n]$: valid cover array.
\item $B[1\twodots n]$: border array keeping track of the lengths of longest borders.
\item $x[1\twodots n]$: string constructed by the algorithm.
\end{itemize}

We solve the stated problem in three steps.

\begin{itemize}
\item{\itshape Step} 1~({\itshape Array Transformation}): Adopting the same strategy used in~\cite{DBLP:conf/cpm/CrochemoreIPT10}, convert the input cover array to a minimal cover array $C$ using procedure \proc{MaxToMin}\index{M\proc{axToMin}}~\cite{DBLP:conf/cpm/CrochemoreIPT10} in case a maximal cover array is given as input.
\item{\itshape Step} 2~({\itshape Pruning}): Covert the (minimal) cover array $C$ to a pruned (minimal) cover array $C^{P}$ by applying procedure \proc{Prune}~\cite{DBLP:conf/cpm/CrochemoreIPT10}.
\item{\itshape Step} 3~({\itshape String Inference}):
    \begin{enumerate}
        \item[i)] Construct a cover graph $G(V,E)$ from $C^{P}$. This graph $G$ has the same connected components as the graph directly constructed from $C$~\cite{DBLP:conf/cpm/CrochemoreIPT10}.
        \item[ii)] Compute connected components of $G$. Decide which character to assign to the first position of each component as follows : Let, $i$ be the first position of any component. If $x[B[i-1]+1]\ =\ a $, then assign $b$ to $x[i]$. Otherwise, assign $a$ to $x[i]$.
    \end{enumerate}
 For each position $j$ in string $x$, the algorithm also computes $B[j]$ online, according to the well-known ``Failure Function Algorithm'' described in~\cite{DBLP:books/aw/AhoHU74}.
\end{itemize}

The procedure \proc{MaxToMin}\index{M\proc{axToMin}} described in \cite{DBLP:conf/cpm/CrochemoreIPT10}, works as follows:

Given, a cover array $C[1\twodots n]$, it checks each value $C[i]$, $1\leq i \leq n$ as follows.
\begin{itemize}
\item[-] if $C[i] = 0$, then leaves it unchanged.
\item[-] if $C[i] \neq 0$, then substitutes $C[i]$ with $C[C[i]]$, provided $C[C[i]]$ is \emph{nonzero}.
Otherwise, $C[i]$ is kept unchanged.	
\end{itemize}

The procedure \proc{Prune}\index{P\proc{rune}} described in \cite{DBLP:conf/cpm/CrochemoreIPT10}, works as follows:

Given, a cover array $C[1\twodots n]$, it finds each totally covered position $i$ and substitutes $C[i]$ by $0$.
\begin{itemize}
\item[-] The procedure scans $C[1\twodots n]$ from large to small indices.
\item[-] Keeps a variable $l$, initially made $0$. If at any instant $l$ is larger than $C[i]$, then $i$ is a totally covered position. So, $C[i]$ is made $0$.
\item[-] At each iteration, next value of $l$ is computed.
\end{itemize}

For ease of understanding the procedures \proc{MaxToMin} and \proc{Prune} are given in Figure~\ref{fig:mtm} and Figure~\ref{fig:prune}. The algorithm
\proc{SIMA} is given in Figure~\ref{fig:sima2}. And its execution steps for a given cover array is illustrated in Figure~\ref{fig:simul2}.

\begin{figure}[!htbp]
\centering
\begin{algorithm}
{SIMA}{C, n}
  C \= \CALL{MaxToMin}(C, n); \\
  C \= \CALL{Prune}(C, n); \\
  \triangleright \text{{\itshape Produce Edges}}\\
  \begin{FOR}{i \= 1 \TO n}\\
  	E[i] \= empty\ list;
  \end{FOR}\\
  \begin{FOR}{i \= 1 \TO n}
    \begin{FOR}{j \= 1 \TO C[i]}
      E[i-C[i]+1+j].add(j);\\
      E[j].add(i-C[i]+1+j);
    \end{FOR}
  \end{FOR}\\
  \triangleright \text{{\itshape Compute connected components by DFS and assign characters to output string}}\\
  S \= empty\ stack; \\
  ch \=\ \text{`$a$'};\\
  \begin{FOR}{i \= 1 \TO n}\\
    \begin{IF}{x[i] = \NIL}
      \CALL{S.push}(i);\\
      \begin{IF}{i > 1 \text{and} C[i] = 0}
        \begin{IF}{x[B[i-1] + 1] =\ \text{`$a$'}}
          ch \=\ \text{`$b$'};
        \ELSE
          ch \=\ \text{`$a$'};
        \end{IF}
      \end{IF}\\
      \begin{WHILE}{not\ \CALL{S.empty}()}
        p \= \CALL{S.pop}();\\
        x[p] \= ch;\\
        \begin{FOR}{each\ element\ j\ of\ E[p]}\\
          \begin{IF}{x[j] = \NIL}
            \CALL{S.push}(j);
          \end{IF}
        \end{FOR}
      \end{WHILE}
    \end{IF}\\
   \begin{IF}{i>1}
    l \= B[i-1]+1;\\
    \begin{WHILE}{l \neq 0}\\
      \begin{IF}{x[i]=x[l]}
        B[i] \= l;\\
        Break;
       \ELSE
         l \= B[l-1] + 1;
      \end{IF}
    \end{WHILE}\\
    \begin{IF}{l = 0}
    \begin{IF}{x[i] = x[1]}
      B[i]\= 1;
      \ELSE
       B[i]\= 0;
     \end{IF}
    \end{IF}\\
  \end{IF}
  \end{FOR}\\
  \RETURN x;
\end{algorithm}
\caption{Algorithm \proc{SIMA}.}
\label{fig:sima2}
\end{figure}

\begin{figure}[!htbp]
\begin{algorithm}
{MaxToMin}{C,n}
\begin{FOR}{i \gets 1 \TO n}\\
	\begin{IF}{C[i] \neq 0\ and\ C[C[i]] \neq 0}
		C[i] \= C[C[i]]
	\end{IF}
\end{FOR}
\end{algorithm}
\caption{Procedure \proc{MaxToMin}.}
\label{fig:mtm}
\end{figure}

\begin{figure}[!htbp]
\begin{algorithm}
{Prune}{C,n}
l \= 0\\
\begin{FOR}{i \= n \TO 0}\\
    \begin{IF}{l \geq C[i]}
        C[i] \= 0
    \end{IF}\\
    l \= \CALL{Max}(0, \CALL{Max}(l, C[i])-1)
\end{FOR}\\
\RETURN C
\end{algorithm}
\caption{Procedure \proc{Prune}.}
\label{fig:prune}
\end{figure}

\begin{figure}[!htbp]
\begin{center}
\begin{tabular}{|c|c|c|c|c|c|c|c|c|c|c|c|c|c|}
  \hline
  $i$   &1 &2 &3 &4 &5 &6 &7 &8 &9 &10 &11 &12 &13\\
  \hline
  $C[i]$ &0 &1 &0 &0 &0 &0 &0 &0 &0 &0 &0  &6 &0\\
  \hline
  $C^{P}[i]$ &0 &1 &0 &0 &0 &0 &0 &0 &0 &0  &0 &6 &0\\
  \hline
\end{tabular}
\end{center}
\begin{center}{\bf (a)}\end{center}
\begin{center}
\begin{tabular}{ccccccc}
Components:& $\{1, 2, 7, 8\}$ & $\{3, 9\}$ & $\{4, 10\}$ & $\{5, 11\}$ &  $\{6, 12\}$ & $\{13\}$
\end{tabular}
\end{center}
\begin{center}{\bf (b)}\end{center}
\begin{center}
\begin{tabular}{|c|c|c|c|c|c|c|c|c|c|c|c|c|c|c|}
\hline
   $i$   &1 &2 &3 &4 &5 &6 &7 &8 &9 &10 &11 &12 &13 & Comment\\
\hline
$x[i]$ &a &a &  &  &  &  &a &a &  &   &   &   &   & $x[1] \leftarrow a \  and \ x[2],\  x[7],\  x[8]$ \\
$B[i]$ &0 &  &  &  &  &  &  &  &  &   &   &   &   & $\ are\  assigned\  x[1]$\\
\hline
\hline
$x[i]$ &a &a  &  &  &  &  &a &a &  &   &   &   &   &\\
$B[i]$ &0 &1  &  &  &  &  &  &  &  &   &   &   &   &\\
\hline
\hline
$x[i]$ &a &a  &b  &  &  &  &a &a &b  &   &   &   &   & $x[B[2] + 1]\  is\  a.\ So\  x[3] \leftarrow b.$\\
$B[i]$ &0 &1  &0  &  &  &  &  &  &   &   &   &   &   & $\ x[9] \leftarrow x[3]$\\
\hline
\hline
$x[i]$ &a &a  &b  &b  &  &  &a &a &b  &b   &   &   &   & $x[B[3] + 1]\  is\  a.\ So\  x[4] \leftarrow b.$\\
$B[i]$ &0 &1  &0  &0  &  &  &  &  &   &   &   &   &   &$\ x[10] \leftarrow x[4]$\\
\hline
\hline
$x[i]$ &a &a  &b  &b  &b  &  &a &a &b  &b   &b   &   &   &$x[B[4] + 1]\  is\  a.\ So\  x[5] \leftarrow b.$\\
$B[i]$ &0 &1  &0  &0  &0  &  &  &  &   &   &   &   &   &$\ x[11] \leftarrow x[5]$\\
\hline
\hline
$x[i]$ &a &a  &b  &b  &b  &b  &a &a &b  &b   &b   &b   &   &$x[B[5] + 1]\  is\  a.\ So\  x[6] \leftarrow b.$\\
$B[i]$ &0 &1  &0  &0  &0  &0  &  &  &   &   &    &   &   &$\ x[12] \leftarrow x[6]$\\
\hline
\hline
$x[i]$ &a &a  &b  &b  &b  &b  &a &a &b  &b   &b   &b   &   &\\
$B[i]$ &0 &1  &0  &0  &0  &0  &1  &  &   &   &    &   &   &\\
\hline
\hline
$x[i]$ &a &a  &b  &b  &b  &b  &a &a &b  &b   &b   &b   &   &\\
$B[i]$ &0 &1  &0  &0  &0  &0  &1  &2  &   &   &    &   &   &\\
\hline
\hline
$x[i]$ &a &a  &b  &b  &b  &b  &a  &a    &b  &b   &b   &b   &  &\\
$B[i]$ &0 &1  &0  &0  &0  &0  &1  &2    &3   &   &    &   &   &\\
\hline
\hline
$x[i]$ &a &a  &b  &b  &b  &b  &a  &a  &b  &b    &b   &b   &   &\\
$B[i]$ &0 &1  &0  &0  &0  &0  &1  &2  &3   &4   &    &   &   &\\
\hline
\hline
$x[i]$ &a &a  &b  &b  &b  &b  &a  &a  &b   &b   &b   &b   &   &\\
$B[i]$ &0 &1  &0  &0  &0  &0  &1  &2  &3   &4   &5    &   &   &\\
\hline
\hline
$x[i]$ &a &a  &b  &b  &b  &b  &a  &a  &b   &b   &b   &b   &   &\\
$B[i]$ &0 &1  &0  &0  &0  &0  &1  &2  &3   &4   &5   &6   &   &\\
\hline
\hline
$x[i]$ &a &a  &b  &b  &b  &b  &a  &a  &b   &b   &b   &b   &b   &$x[B[12] + 1]\  is\  a.\ So\  x[13] \leftarrow b.$\\
$B[i]$ &0 &1  &0  &0  &0  &0  &1  &2  &3   &4   &5   &6   &0   &\\
\hline
\end{tabular}
\end{center}
\begin{center}{\bf (c)}\end{center}
\caption[An example run of Algorithm \proc{SIMA}.]{An example run of Algorithm \proc{SIMA}. (a) Input cover array $C$ before and after pruning, (b) Connected components of corresponding cover graph, and (c) String Inference by Algorithm \proc{SIMA}.}
\label{fig:simul2}
\end{figure}

Now, we state and prove the main findings.

\begin{thm}
\label{thm:s21}
Let $C^P$ be a pruned cover array of input cover array $C$, which resulted from Step 2 of the Algorithm \proc{SIMA}. Let $x$ be the word which is a result of the Algorithm \proc{SIMA}. Let $C_x$ is the (minimal) cover array for $x$. Then $C=C_x$.
\end{thm}

\begin{proof}
\label{prf:s21}
We just need to show that each assignment of a character to position $i$ of the string $x$ does not violate any constraints set by the values of $C^P[i]$.

Here we first construct the cover graph $G$ from $C^{P}$. Then the nonzero values in $C^{P}$ state
that, the letters at positions $i$ and $j$ of $x$ need to be equal, if $i$ and $j$
are connected in $G$. Since pruning does not reduce vertex connectivity~\cite{DBLP:conf/cpm/CrochemoreIPT10}, the cover
graph induced by $C^{P}$ has the same connected components as the one induced by
$C$. The number of edges in the graph induced by $C^{P}$ is bounded by
$2n$ according to Property~\ref{prop:cover6}~\cite{DBLP:conf/cpm/CrochemoreIPT10}.

After constructing the graph, we compute the connected components of
the constructed graph and at the same time assigns characters to the output string
$x$. It also computes the value of longest border $B[i]$ for string $x[1\twodots i]$ for
each $i$ as the iterations advances. Computation of connected component is done to assign same character to those positions
in the string which correspond to member vertices of a connected component.

We take decision only to assign a character to  the first member~(from left)
of a component, and assign the same character to the remaining members of that
component. That is, we can consider the following two cases:

\begin{enumerate}
    \item When $C^{P}[i] = k, 0<k<i$. This means, $i$ has an edge with $k$, hence both $i$ and $k$ belong to the same component. So, whenever a character is assigned to $x[k]$ it is also assigned to $x[i]$. Thus we do not need to take a decision about which character to assign to $x[i]$ when   $C^{P}[i]$ is nonzero.
    \item When $C^P[i] = 0$. If position $i$ corresponds to the first member of a component, we check the value of $B[i-1]$. Let, $B[i-1] = k$. We can satisfy $C^{P}[i] = 0$, if we can ensure  $B[i] \leq  B[i-1]$, as stated in Property~\ref{prop:cover7}~\cite{DBLP:journals/algorithmica/LiS02}. Thus we assign $x[i]$ a character different from $x[k+1]$ so that no border of length greater than $k$ is possible for $x[1\twodots i]$. This obviously keeps $C^{P}[i] = 0$. Again, if $i$ does not correspond to the first member of a component, then it is already assigned a valid character according to the component condition (i.e., all other members of a component receive the same character as the first one).
\end{enumerate}

Thus the resultant string $x[1\twodots n]$ satisfies the pruned cover array $C^P$ at every position.
\end{proof}

\begin{thm}
\label{thm:s22}
Any string constructed by the algorithm \proc{SIMA} uses an alphabet comprising no more than two characters.
\end{thm}

\begin{proof}
\label{prf:s22}
We prove this claim by induction on the length of cover array.

Without loss of generality, let, $C[1\twodots n]$ be a valid (minimal) cover array
of string $x$ of length $n$. Let, the two characters to be assigned to infer the output string $x$ be in $\{a, b\}$.

{\bf Basis:} When $n = 1$, for a valid cover array $C[1] = 0$. In this case, $x$ constitutes of a single character {\itshape `a'} and $B[1] = 0$.

When $n = 2$, two values of $C[2]$ are possible for a valid cover array $C$. One is $C[2] = 1$. In this case,  $x[2]$ must be {\itshape `a'} to obtain $x = aa$ . Otherwise, $C[2] = 0$. In this case, $x[2]$ must be {\itshape `b'}  to obtain $x = ab$ . In both case, value of $B[2]$ is computed.

{\bf Induction:} Let $n>2$. We assume that up to length $n-1$,
$B[1\twodots n-1]$ and $x[1\twodots n-1]$ have been computed and $x[1\twodots n-1]$
needs an alphabet consisting of two characters. We consider the assignment of
character to $x[n]$.

{\itshape Case 1}: $C[n] = 0$

According to Property~\ref{prop:cover7}, for every integer $1 \leq i \leq n-1$,
if $B[i+1] <= B[i]$ then $C[i+1] = 0$.

Let, $B[n-1] = k$. Now, if $x[k+1] = $ {\itshape `a'} then we assign
{\itshape `b'}  to $x[n]$ so that $B[i+1]$ cannot become greater than $k$.
Or, if $x[k+1] = $ {\itshape `b'}  then we assign {\itshape `a'}  to $x[n]$
for the same reason. This maintains the constraint $C[n] = 0$.
So $x[1..n]$ uses a two-character alphabet.

{\itshape Case 2}: $C[n] = k,\ 1 \leq k < n$

Position $n$ has an edge with position $k$. Our algorithm assigns
into $x[n]$ the same character that it assigns into $x[k]$. Since $k < n$,
so $x[k]$ is either {\itshape `a'} or {\itshape `b'}. Thus, we do not need to
introduce any new character for $x[n]$ here.

Thus algorithm \proc{SIMA} produces a string $x[1\twodots n]$ which uses an
alphabet of no more than two characters.
\end{proof}

\begin{thm}
\label{thm:s23}
Algorithm \proc{SIMA}\index{S\proc{IMA}!complexity} runs in linear time.
\end{thm}

\begin{proof}
\label{prf:s23}
The each of the two procedures \proc{MaxToMin}~\cite{DBLP:conf/cpm/CrochemoreIPT10} and \proc{Prune}~\cite{DBLP:conf/cpm/CrochemoreIPT10} runs in linear
time~\cite{DBLP:conf/cpm/CrochemoreIPT10}. The step of producing edges $E$ of graph $G$ induced by $C^{P}$ is also linear because the number of edges is bounded by $2n$ according to Property~\ref{prop:cover5}~\cite{DBLP:conf/cpm/CrochemoreIPT10}.

The third for loop computes the connected components in the graph by depth first search and assigns letters to the output string. This computation is linear in
the number of edges which is bounded by $2n$. Also the overall on-line computation of the border array $B$ runs in linear time~\cite{DBLP:journals/siamcomp/KnuthMP77}. Hence our algorithm runs in linear time.
\end{proof}

\section{Experimental Results}
\label{sec:exp}
We have investigated the practical performance of Algorithm \proc{SIMA}
on various datasets. The experiments were performed on a computer with \emph{4 GB}
of main memory and {\itshape 3.1 GHz Intel Pentium 4} processor, running the
{\itshape Windows XP Service Pack 3} operating system. All programs were compiled
with {\itshape Visual Studio 6.0}.

The investigated data includes, all valid cover arrays for length $8$ to $14$
and cover arrays generated from {\itshape Fibonacci words} of different sizes.
The experimental results are summarized below.

\begin{itemize}
\item We have been able to verify the linear runtime of our algorithm
experimentally. Figure~\ref{fig:graph} shows the timing diagram of our algorithm for {\itshape fibonacci word dataset}. For hardware limitations we restricted our test from {\itshape fibonacci word} size $4$ to $34$.

  \begin{figure}[!htbp]
    \begin{center}
     \includegraphics{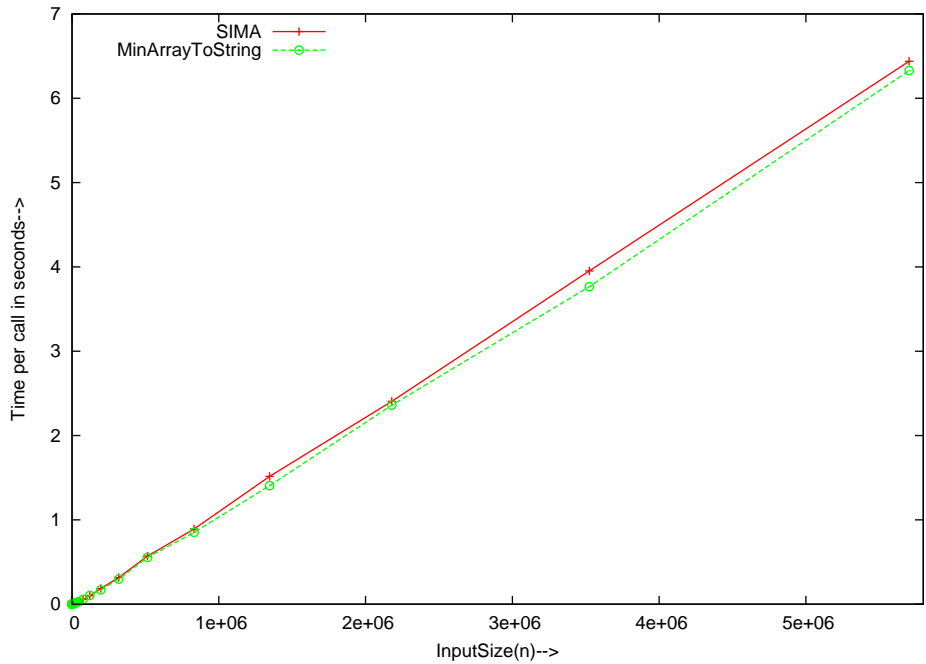}
      \caption{Verification of Linear runtime of Algorithm \proc{SIMA}.}
      \label{fig:graph}
    \end{center}
  \end{figure}

\item We have also compared our algorithm with the implementation of
\proc{MinArrayToString}\index{M\proc{inArrayToString}} available at~\cite{KCL}.
In every case, our algorithm was able to infer valid strings with no more than two letters which is a sure improvement over \proc{MinArrayToString}. The comparative results for all valid cover arrays of length $8$ is shown in Table~\ref{comp}. Table ~\ref{comp2} shows the comparison of the two algorithms for several genome sequences available at~\cite{WISC}.

\begin{table}
\begin{center}
\begin{tabular}{|c|c|c|}
\hline
Input Cover Array & String Inferred By & String Inferred By\\
& \proc{SIMA} & \proc{MinArrayToString}\\
\hline
\hline
0 0 0 0 0 0 0 0 & \emph{a b b b b b b b} & \emph{a b c d e f g h}\\
\hline
0 0 0 0 0 0 0 4 & \emph{a b b b a b b b} & \emph{a b c d a b c d}\\
\hline
0 0 0 0 0 3 0 0 & \emph{a b b a b b b b} & \emph{a b c a b c d e}\\
\hline
0 0 0 0 0 3 0 3 & \emph{a b a a b a b a} & \emph{a b a a b a b a}\\
\hline
0 0 0 0 0 3 4 0 & \emph{a b b a b b a a} & \emph{a b c a b c a d}\\
\hline
0 0 0 0 0 3 4 5 & \emph{a b b a b b a b} & \emph{a b c a b c a b}\\
\hline
0 0 0 2 0 0 0 0 & \emph{a b a b b b b b} & \emph{a b a b c d e f}\\
\hline
0 0 0 2 3 0 0 0 & \emph{a b a b a a a a} & \emph{a b a b a c d e}\\
\hline
0 0 0 2 3 0 0 3 & \emph{a b a b a a b a} & \emph{a b a b a a b a}\\
\hline
0 0 0 2 3 2 0 0 & \emph{a b a b a b b b} & \emph{a b a b a b c d}\\
\hline
0 0 0 2 3 2 3 0 & \emph{a b a b a b a a} & \emph{a b a b a b a c}\\
\hline
0 0 0 2 3 2 3 2 & \emph{a b a b a b a b} & \emph{a b a b a b a b}\\
\hline
0 1 0 0 0 0 0 0 & \emph{a a b b b b b b} & \emph{a a b c d e f g}\\
\hline
0 1 0 0 0 0 0 4 & \emph{a a b b a a b b} & \emph{a a b c a a b c}\\
\hline
0 1 0 0 0 3 0 0 & \emph{a a b a a b b b} & \emph{a a b a a b c d}\\
\hline
0 1 0 0 0 3 4 0 & \emph{a a b a a b a b} & \emph{a a b a a b a c}\\
\hline
0 1 0 0 0 3 4 5 & \emph{a a b a a b a a} & \emph{a a b a a b a a}\\
\hline
0 1 1 0 0 0 0 0 & \emph{a a a b b b b b} & \emph{a a a b c d e f}\\
\hline
0 1 1 0 0 0 0 4 & \emph{a a a b a a a b} & \emph{a a a b a a a b}\\
\hline
0 1 1 1 0 0 0 0 & \emph{a a a a b b b b} & \emph{a a a a b c d e}\\
\hline
0 1 1 1 1 0 0 0 & \emph{a a a a a b b b} & \emph{a a a a a b c d}\\
\hline
0 1 1 1 1 1 0 0 & \emph{a a a a a a b b} & \emph{a a a a a a b c}\\
\hline
0 1 1 1 1 1 1 0 & \emph{a a a a a a a b} & \emph{a a a a a a a b}\\
\hline
0 1 1 1 1 1 1 1 & \emph{a a a a a a a a} & \emph{a a a a a a a a}\\
\hline
\end{tabular}
\end{center}
\caption{Comparison on {\itshape Inferred String} between algorithms \proc{SIMA}
and \proc{MinArrayToString}.}
\label{comp}
\end{table}

\begin{table}
\begin{center}
\begin{tabular}{|c|c|c|}
\hline
Genome Sequence & \proc{SIMA} & \proc{MinArrayToString}\\
\hline
\hline
\emph{Acidovorax citrulli} AAC00-1	& 2 &  5352783\\
\hline		
\emph{Buchnera aphidicola} 5A		& 2 &  642133\\
\hline		
\emph{Ca. Blochmannia floridanus}	& 2 &  705649\\
\hline		
\emph{Dickeya dadantii} 3937		& 2 &  4922813\\
\hline  		
\emph{Edwardsiella ictarluri} 93-146	& 2 &  3812326\\
\hline		
\emph{Klebsiella pneumonia} 342	& 2 &  5920281\\		
\hline
\end{tabular}
\end{center}
\caption{Comparison on {\itshape Alphabet Size} between algorithms \proc{SIMA} and \proc{MinArrayToString}.}
\label{comp2}
\end{table}

\item Finally we have observed an interesting fact that the set of distinct valid cover arrays is generated from $m$-alphabet string for a certain length, where $m \geq 2$. We generated all possible strings for length of 8 with alphabet sizes 2, 3, 4, 5, 6, 7 and 8, and computed cover arrays for all of them. For each alphabet size we got same set of distinct cover arrays\index{S\proc{IMA-II}|)}.
\end{itemize}

\section{Conclusion}
\label{sec:con}

In this paper, we have presented a linear time algorithm to solve the problem of inference of strings using the least sized alphabet~(i.e., binary alphabet) from valid cover arrays. We achieved the least possible bound on alphabet size by incorporating an interesting relation between cover array and border array of a string. The main finding of this paper is that, from any valid cover array of length $n$, it is possible to infer a string over an alphabet that consists only two distinct characters unless the cover array is of the form $01^{k-1},\ 1\leq k \leq n$. In that particular case, our algorithm infers a string over an alphabet consisting only of a single character.

\raggedright
\bibliographystyle{ieeetr}
\bibliography{refCover}

\end{document}